\documentclass[journal]{IEEEtran}
\usepackage{color}

\ifCLASSINFOpdf
\else
\fi
\usepackage[cmex10]{amsmath}
\usepackage{amssymb}
\usepackage{amsthm}
\begin{document}
%
\title{On zero neighbours and trial sets of linear codes}
%
%
%

\author{Mijail Borges-Quintana, \and~{Miguel~\'Angel Borges-Trenard}, \and~Edgar Mart\'inez-Moro
\thanks{M.~Borges-Quintana and M.A.~Borges-Trenard are with the Department
of Department of Mathematics, Faculty of Mathematics and Computer Science, Universidad de Oriente, Santiago de Cuba, Cuba. \tt{{mijail@csd.uo.edu.cu, mborges@csd.uo.edu.cu}}.}
\thanks{E.~Mart\'inez-Moro  is with the Institute of Mathematics IMUVa, University of Valladolid. Valladolid, Castilla, Spain. \tt{edgar@maf.uva.es}}
\thanks{Manuscript received --; revised  --.}}

%
%

\markboth{ }%
{Shell \MakeLowercase{\textit{et al.}}: Bare Demo of IEEEtran.cls for Journals}
%



\theoremstyle{theorem}
\newtheorem{theorem}{Theorem} 
\newtheorem{proposition}{Proposition} 
\theoremstyle{definition}
\newtheorem{definition}{Definition} 

\theoremstyle{remark}
\newtheorem{remark}{Remark} 

\maketitle

\begin{abstract}
In this work we   study the set
of  leader codewords of a non-binary linear code. This set has some nice properties related to the  monotonicity  of the weight compatible order on the generalized support of a
vector in $\mathbb F_q^n$. This allows us to describe a test set,  a trial set and the set zero neighbours in terms of the leader codewords.
\end{abstract}


%
\IEEEpeerreviewmaketitle

\section{Introduction}
%
%
%
%
\IEEEPARstart{A}s it is pointed in \cite{klove} it is common folklore in the theory
of binary codes that there is an ordering on the coset leaders chosen as
the lexicographically smallest minimum weight vectors that provides a
monotone structure. In the binary case this is expressed as follows: if
$\mathbf x$ is a coset leader  and $\mathbf y\subseteq \mathbf x$ (i.e.
$y_i\leq x_i$ for all $i$)  then $\mathbf y$ is also a coset leader.
This nice property has been proved of great value, see for example
\cite{Zemor}, and it has been used for analyzing the error-correction
capability of binary linear codes \cite{klove}. In this last paper
the authors introduce  the concept of a \textit{trial set} of codewords
and they provide a gradient-like decoding algorithm based on this set.

Despite the interest of this topic no generalization of this ideas is
known by the authors of this communication to the $q$-ary linear case,
that is, to codes over $\mathbb F_q$ the field of $q=p^r$ elements where
$p$ is a prime. In this paper we provide a non straightforward
generalization of this ideas to the $q$-ary case. In
Section~\ref{sec:Prel} we introduce the idea of a generalized support of a
vector in $\mathbb F_q^n$ based on considering $\mathbb F_q$ as a
$\mathbb F_p$ vector space using the so called $p$-ary expansion. Section~\ref{sec:LCodW} is devoted to the study
of the leader codewords of a code as a zero neighbor set and their
properties.  The leader codewords have been previously defined for the binary case in \cite{borges:2014} but the analogous concept in the $q$-ary case needs a subtle and non-trivial generalization as well as their properties. Finally in Section~\ref{Tset-LCW} we analyse the correctable and uncorrectable errors defining a trial set for a linear code from the set of leader codewords.

\section{Preliminaries}\label{sec:Prel}

From now on we  shall denote by  $\mathbb F_q$  the finite field with $q$ elements, with $q=p^m$ and $p$ a prime. A \emph{linear code} $\mathcal C$ over $\mathbb F_q$ of length $n$ and dimension $k$, or an $[n,k]$ linear code for short, is a $k$-dimensional subspace of $\mathbb F_q^n$. We will call the vectors $\mathbf v$ in $\mathbb F_q^n$ words and the particular case where $\mathbf v\in \mathcal C$, codewords. For every word $\mathbf v\in \mathbb F_q^n$ its \emph{support} is define as its support as a vector in $\mathbb F_q^n$, i.e. $\mathrm{supp}(\mathbf v) = \left\{ i \mid v_i \neq 0\right\}$ and its \emph{Hamming weight}, denoted by $\mathrm{w_H}(\mathbf v)$ as the cardinality of $\mathrm{supp}(\mathbf v)$. 

As usual the \emph{Hamming distance}, $\mathrm{d_H}(\mathbf x, \mathbf y)$, between two words $\mathbf x, ~\mathbf y \in \mathbb F_q^n$ is the number of places where they differ, or equivalently, 
$\mathrm{d_H}(\mathbf x, \mathbf y) = \mathrm{w_H}(\mathbf x - \mathbf y)$. The \emph{minimum distance} $\mathrm{d}(\mathcal C)$ of a linear code $\mathcal C$ is defined as the minimum weight among all nonzero codewords. 

The words of minimal Hamming weight in the cosets of $\mathbb F_q^n / \mathcal C$ is the \textit{set of coset leaders} of the code $\mathcal C$ in $\mathbb F_q^n$ and we will denote it by $\mathrm{CL}(\mathcal C)$.  $\mathrm{CL}(\mathbf y)$ will denote the subset of coset leaders corresponding to the coset $\mathbf y+\mathcal C$.
 Given a coset $\mathbf y+\mathcal C$ we define the \emph{weight of the coset} $\mathrm{w_H}(\mathbf y+\mathcal C)$ as the smallest Hamming weight among all vectors in the coset, or equivalently the weight of one of its leaders. It is well known that if  $t= \lfloor \frac{d(\mathcal C)-1}{2}\rfloor$ is the \emph{error-correcting capacity} of $\mathcal C$ where $\lfloor \cdot \rfloor$ denotes the greatest integer function then every coset of weight at most $t$ has a unique coset leader.

 We will asume  that $\mathbb F_q= \frac{\mathbb F_p[X]}{\left(f(X)\right)}$ where $f(X)$ is chosen such that $f(X)$ is an irreducible polynomial over $\mathbb F_p$ of degree $m$. Let $\beta$ be a root of $f(X)$, then an equivalent representation of $\mathbb F_q$ is $\mathbb F_p[\beta]$, i.e. any element of $a \in \mathbb F_q$ is represented as 
$$a_1 + a_2 \beta + \ldots + a_{m}\beta^{m-1} \hbox{ with }
a_i \in \mathbb F_p \hbox{ for } i \in \{1, \ldots, m\}.$$
For a word $\mathbf v=(v_1, \ldots, v_n) \in \mathbb F_q^n$, such that the $i$-th component of $\mathbf v$ is
$$v_i = a_{i,1} + a_{i,2} \beta + \ldots + a_{i,m}\beta^{m-1}$$
we define the \emph{generalized support} of a vector $\mathbf v$  as the support of the $nm$-tuple given by the concatenations of the $p$-adic expansion of each component $\mathbf v_i$ of $\mathbf v$, i.e.{
$$\mathrm{supp}_{\mathrm{gen}}(\mathbf v) = \{\mathrm{supp}((a_{i,0}, \ldots, a_{i,m-1}))\}:\, i=1\ldots n],$$
and $\mathrm{supp}_{\mathrm{gen}}(\mathbf v)[i]=\mathrm{supp}((a_{i,0}, \ldots, a_{i,m-1}))$.}

We will say that $(i,j)\in \mathrm{supp}_{\mathrm{gen}}(\mathbf v)$ if the corresponding $a_{i,j}$ is not zero. From now on the set $\left\{ \mathbf e_{ij}=\beta^{j-1}\mathbf e_i:\; i = 1, \ldots, n;\, j= 1,\ldots m \right\}$ will be denoted as $\mathrm{Can}(\mathbb F_q,f)$ and it represents the canonical basis of {$(\mathbb F_q^n,+)$, the additive monoid $\mathbb F_q^n$ with respect to the ``+" operation, where $f$ is the irreducible polynomial used to define $\mathbb F_q^n$}.

{We will say that a representantion of a word $\mathbf w$ as an $nm$-tuple is in \textit{standard form} if it can not be reduced with respect to the ``+" operation in the representation, i.e., it is the minimal representation for $\mathbf w$ {with respect to the generalized support}. We will denote the standard form of $\mathbf v$ as $\mathrm{SF}(\mathbf v,f)$. Therefore, $\mathbf w$ is in stardard form if $\mathbf w \equiv \mathrm{SF}(\mathbf w,f)$ (we will also say $\mathbf w \in \mathrm{SF}(\mathbb F_q^n,f)$), note that we do not use common equality symbol in order to emphasize that the equality is considered as literal expresions in terms of the basis $\mathrm{Can}(\mathbb F_q,f)$.}

{\begin{remark} From now on we will use  $\mathrm{Can}(\mathbb F_q)$ and $\mathrm{SF}(\mathbb F_q^n)$ instead of $\mathrm{Can}(\mathbb F_q,f)$ and $\mathrm{SF}(\mathbb F_q^n,f)$ respectively since it is clear that different elections of $f$ or $\beta$ provide equivalent generalized supports.
\end{remark}}

{Given $\mathbf x,\mathbf y\in (\mathbb F_q^n,+)$, $\mathbf x=\sum_{i,j}x_{ij}\mathbf e_{ij}$, $\mathbf y=\sum_{i,j}y_{ij}\mathbf e_{ij}$, $\mathbf x\subset \mathbf y$ if $x_{ij}\leq y_{ij}$ for all $i\in [1,n]$, and $j\in[1,m]$. An \textit{admisible order} on $(\mathbb F_q^n,+)$ is a total order $<$ on $(\mathbb F_q^n,+)$ satisfying the following two conditions
\begin{enumerate}
\item $\mathbf 0 <\mathbf x$, for all $\mathbf x\in (\mathbb F_q^n,+),\,\mathbf x\neq \mathbf 0$.
\item If $\mathbf x<\mathbf y$, then $\mathbf x \oplus \mathbf z<\mathbf y \oplus \mathbf z$, for all $\mathbf z\in (\mathbb F_q^n,+)$, where $\oplus$ denotes the operation of grouping the same canonical words of $\mathrm{Can}(\mathbb F_q)$ but without reducing the coefficients module $p$.
\end{enumerate}
}

\begin{definition}
\label{Voronoi:Definition}
We define the  \emph{Voronoi region} of a codeword $\mathbf c \in \mathcal C$ and we denote it by $\mathrm D(\mathbf c)$ as the set
\begin{equation*}\mathrm D(\mathbf c) = \left\{ \mathbf y \in \mathbb F_q^n \mid
\mathrm{d_H}(\mathbf y, \mathbf c)\leq \mathrm{d_H}(\mathbf y, \mathbf c'),\,\forall \mathbf c'\in \mathcal C\setminus \{ \mathbf 0 \} \right\}.\end{equation*}
\end{definition}
Note that the set of all the Voronoi regions for a given linear code $\mathcal C$ covers the space $\mathbb F_q^n$ and also it is clear that $\mathrm D(\mathbf 0) = \mathrm{CL}(\mathcal C)$. However, some words in $\mathbb F_q^n$ may be contained in several regions.

For any subset $A \subset \mathbb F_q^n$ we define $\mathcal X(A)$ as the set of words at Hamming distance $1$ from $A$, i.e.
$$\mathcal X (A) = \left\{ \mathbf y \in \mathbb F_q^n \mid \min \left\{ \mathrm{d_H}(\mathbf y, \mathbf a)~:~ \mathbf a \in A\right\} = 1\right\}.$$
We define the \emph{boundary} of $A$ as $\delta (A) = \mathcal X (A) \cup \mathcal X (\mathbb F_2^n \setminus A)$.

\begin{definition}
A nonzero codeword $\mathbf c \in \mathcal C$ is called a \emph{zero neighbour} if its Voronoi region shares a common boundary with the set of coset leaders, i.e. 
$$\delta (\mathrm D(\mathbf z)) \cap \delta (\mathrm D(\mathbf 0)) \neq \emptyset.$$
We will denote by $\mathcal Z(\mathcal C)$ the set of all zero neighbours of $\mathcal C$
$$\mathcal Z(\mathcal C) = \left\{ \mathbf z \in \mathcal C \setminus \{ \mathbf 0\}
~:~ \delta (\mathrm D(\mathbf z)) \cap \delta (\mathrm D(\mathbf 0)) \neq \emptyset
\right\}.$$
\end{definition}

\begin{definition}
A \emph{test-set} $\mathcal T$ for a given linear code $\mathcal C$ is a set of codewords such that every word $\mathbf y$ \begin{enumerate}
  \item either $\mathbf y$ lies in $\mathrm D(\mathbf 0)$, the Voronoi region of the all-zero vector, 
  \item or there exists $\mathbf v\in \mathcal T$ such that $\mathrm{w_H}(\mathbf y- \mathbf v)< \mathrm{w_H}(\mathbf y)$.
\end{enumerate}
\end{definition}

The set of zero neighbours is a test set, also from the set of zero neightbours can be obtained any minimal test set according to the cardinality of the set  \cite{barg:1998}.

\section{Zero neighbours and leader codewords}
\label{sec:LCodW}

The first idea that allow us to  compute incrementally de set of all coset leaders for a linear code was introduced in \cite{borges:2007a}. In that paper we used the additive structure of $\mathbb F_q^n$ with the set of canonical generators $\mathrm{Can}(\mathbb F_q)$. Unfortunatly in \cite{borges:2007a}  most of the chosen coset representatives may not be coset leaders if the weight of the coset is greater than the error-correcting capability of the code.
\begin{theorem}
\label{Theorem1}
{Let $\mathbf w\in \mathrm{SF}(\mathbb F_q^n)$ be an element in  $\mathrm{CL}(\mathcal C)$, and $(i,j)\in \mathrm{supp_{gen}}(\mathbf w)$. Let $\mathbf y\in \mathrm{SF}(\mathbb F_q^n)$ s.t. $\mathbf w = \mathbf y + \mathbf e_{ij}$ then 
$${\mathrm{w_H}(\mathbf y) \leq \mathrm{w_H}(\mathbf y+\mathcal C)+1}.$$}
\end{theorem}
The proof of the Theorem is  analogous to the binary case in  \cite[Theorem 2.1]{borges:2014}. In the situation of  Theorem~\ref{Theorem1} above we will say that the coset leader $\mathbf w$ is an \textit{ancestor} of the word $\mathbf y$, and that $\mathbf y$ is a \textit{descendant} of $\mathbf w$. In the binary case this definitions behave as the ones in \cite[\S11.7]{huffman:2003} but in the  case $q\neq 2$ there is a subtle difference, a coset leader could be an ancestor of a coset leader or an ancestor of a word at Hamming distance $1$ to a coset leader (this last case is not possible in the binary case).

Thus, in order to incrementally generate all coset leaders  starting from $\mathbf 0$ the all zero codeword and adding elements in $\mathrm{Can}(\mathbb F_q)$,  we must consider all words with distance $1$ to a coset leader. Note that to achieve all words at distance one from the coset leaders some  words at distance $2$ are needed.

\begin{theorem}
\label{Theorem2}$\, $
\begin{enumerate}
  \item Let $\mathbf w$ be a word in $\mathbb F_q^n$ such that $\mathrm{d_H}(\mathbf w,\mathrm{CL}(\mathcal C))=1$ and $\mathbf w = \mathbf y + \mathbf e_{ij}$ for some word $\mathbf y \in \mathbb F_q^n$, $(i,j)\in \mathrm{supp_{gen}}(\mathbf w)$ and $\mathbf w, \mathbf y \in \mathrm{SF}(\mathbb F_q^n)$, then 
  \begin{equation}\label{eq:+2}\mathrm{w_H}(\mathbf y) \leq \mathrm{w_H}(\mathbf y+\mathcal C)+2.\end{equation} 
  \item {If equality holds in Eq.~(\ref{eq:+2}) then  $\mathrm{supp_{gen}}(\mathbf y)[i]\neq \emptyset$ and $\mathrm{supp}_{gen}(\mathbf v)[i]=\emptyset$ for any $\mathbf v \in \mathrm{CL}(\mathbf y)$.}
\end{enumerate}
\end{theorem}

Note that in the second case, i.e. when equality in Eq.~(\ref{eq:+2})  holds then the condition implies that 
$\mathrm{d_H}(w,\mathrm{CL}(\mathcal C))=1$. It is clear that if $\mathrm{w_H}(\mathbf y) > \mathrm{w_H}(\mathbf y+\mathcal C)+2$ then $\mathrm{d_H}(\mathbf w,\mathrm{CL}(\mathcal C))>1$ since the addition of a $e_{ij}$ can only modify the distance from $-1$ to $1$.

\subsection{Weight compatible order}

{Given $\prec_1$ an admisible order on $(\mathbb F_q^n,+)$ we define the \textit{ weight compatible order} $\prec$ on  $(\mathbb F_q^n,+)$} associated to  $\prec_1$  as the ordering given by
\begin{enumerate}
  \item  $\mathbf x \prec \mathbf y$ if $\mathrm{w_H}(\mathbf x) < \mathrm{w_H}( \mathbf y)$ or \item if $\mathrm{w_H}( \mathbf x)= \mathrm{w_H}( \mathbf y)$ then $\mathbf x \prec_1 \mathbf y$.
\end{enumerate}
I.e. the words are ordered acording their weights and  $\prec_1$ break ties. This class of orders is a subset of the class  of $\alpha$-orderings monotone in \cite{klove}. In fact we will  need a little more than monotonocity, for the purpose of this work we will also need that for every pair {$\mathbf a, \mathbf b\in \mathrm{SF}(\mathbb F_q^n)$
\begin{equation}
\label{eq:subpal}
\mbox{ if } \mathbf a\subset b, \mbox{ then } \mathbf a \prec \mathbf b.
\end{equation}
}

Note that for any  weight compatible ordering $\prec$  every strictly decreasing sequence  terminates (due to the finiteness of the set $\mathbb F_q^n$) and  the condition in Eq.~(\ref{eq:subpal}) is fulfilled.

\begin{definition}
\label{List-Def-LC}
We define the object $\tt List$ as an ordered set of elements in $\mathbb F_q^n$ w.r.t. a weight compatible order $\prec$ verifying the following properties:
\begin{enumerate}
\item $\mathbf 0\in \tt{List}$.
\item Criterion 1: If $\mathbf v \in \tt{List}$ and $\mathrm{w}_H(\mathbf v) = \mathrm{w}_H\left(N(\mathbf v)\right)$ then $\left\{ \mathbf v + \mathbf e_{ij} \mid \mathbf v + \mathbf e_{ij}\in \mathrm{SF}(\mathbb F_q^n) \right\} \subset \tt{List}$. 
\item Criterion 2: If $\mathbf v \in \tt{List}$ and $\mathrm{w}_H(\mathbf v) = \mathrm{w}_H\left(N(\mathbf v)\right)+1$ then $\left\{ \mathbf v + \mathbf e_{ij} \mid \mathbf v + \mathbf e_{ij}\in \mathrm{SF}(\mathbb F_q^n) \right\} \subset \tt{List}$. 
\item Criterion 3: If $\mathbf v \in \tt{List}$ and $\mathrm{w}_H(\mathbf v) = \mathrm{w}_H\left(N(\mathbf v)\right)+2$ then $\left\{ \mathbf v + \mathbf e_{ij} \mid (i,j)\in I\right\} \subset \tt{List}$, where $I=\{(i,j)\in \mathrm{supp}_{gen}(\mathbf v)\mid \; \mathbf v + \mathbf e_{ij}\in \mathrm{SF}(\mathbb F_q^n),\,\mathrm{supp}_{gen}(\mathbf v^\prime)[i]=\emptyset,\mbox{ for all } \mathbf v^\prime\in \mathrm{CL}(\mathbf v)\}$, 
\end{enumerate}
{where $N(\mathbf v) = \min_{\prec}\left\{ \mathbf w \mid \mathbf w \in \tt{List}\cap \left( \mathcal C+\mathbf v\right)\right\}$.} We denote by $\mathcal N$ the set of distinct $N(\mathbf v)$ with $\mathbf v\in \tt{List}$.
\end{definition}

\begin{remark}
\label{Extra-Remark}
Observe that if $\mathbf v \in \mathbb F_q^n$ satisfies Criterion~1 in Definition~\ref{List-Def-LC} then $\mathbf v \in \mathrm{CL}(\mathcal C)$. In particular, when $\mathbf v$ is the first element of ${\tt List}$ that belongs to $\mathcal C+ \mathbf v$, then $N(\mathbf v) = \mathbf v$.
\end{remark}

\begin{theorem}
\label{Theorem3}
Let $\mathbf w\in \mathbb F_q^n$. If $\mathrm{d_H}(\mathbf w,\mathrm{CL}(\mathcal C))\leq 1$ then $\mathbf w \in \tt{List}$.
\end{theorem}

\begin{proof}
We will proceed by induction on $\mathbb F_q^n$ with the order $\prec$.
The statement is true for $\mathbf 0 \in \mathbb F_q^n$. Now for the inductive step, 
we assume that the desired property is true for any word $\mathbf u \in \mathbb F_q^n$ such tht  $\mathrm{d_H}(\mathbf u,\mathrm{CL}(\mathcal C))\leq 1$ and also $\mathbf u$ is smaller than an arbitrary but fixed $\mathbf w\setminus \{ \mathbf 0\}$ with respect to  $\prec$ and $\mathrm{d_H}(\mathbf w,\mathrm{CL}(\mathcal C))\leq 1$, i.e.
$$\hbox{if } \mathrm{d_H}(\mathbf u,\mathrm{CL}(\mathcal C))\leq 1 \hbox{ and } \mathbf u \prec \mathbf w \hbox{ then }\mathbf u \in \tt{List}.$$
We will show that that the previous conditions imply that $\mathbf w$ is also in $\tt{List}$.

Let $\mathbf w=\mathbf v + \mathbf e_{ij}$, with $(i,j)\in \mathrm{supp_{gen}}(\mathbf w)$ then  $\mathbf v\prec \mathbf w$ by Eq.~(\ref{eq:subpal}). If $\mathrm{d_H}(\mathbf v,\mathrm{CL}(\mathcal C))\leq 1$ then by the induction hypothesis we have that $\mathbf v \in \tt{List}$ and by Criteria 1 or 2 in Definition~\ref{List-Def-LC} it is guaranteed that $\mathbf w\in \tt{List}$. Then, let us suppose that $\mathrm{w_H}(\mathbf v)= \mathrm{w_H}(\mathrm{CL}(\mathbf v))+2$. Since $\mathrm{d_H}(\mathbf w,\mathrm{CL}(\mathcal C))\leq 1$ we have $\mathrm{supp_{gen}}(\mathbf v)[i]\neq \emptyset$ and $\mathrm{supp_{gen}}(\mathbf v^\prime)[i]=\emptyset$ for all $v^\prime\in \mathrm{CL}(\mathbf v)$. Thus we can write $\mathbf v$ as
\begin{equation}
\label{eq:vDcomp}
{\mathbf v}= {\mathbf v}_0 + \sum_{j_k\in \mathrm{supp_{gen}}(\mathbf v)[i]} \mathbf e_{ij_k}, \hbox{ and }\mathrm{supp_{gen}}(\mathbf v_0)[i]=\emptyset.
\end{equation}

Let $L=|\mathrm{supp}_{gen}(\mathbf v)[i]|$ and let us renumber the partial sums of the words  in the previous equation as $\mathbf v_l={\mathbf v}_0 + \sum_{k=1}^{l} \mathbf e_{ij_k(l)}$, where $l=1,\ldots,L$. Note that $\mathbf v_L=\mathbf v$. 
It can be proved that 
\begin{enumerate}
  \item either there exists an element $k\in \{1,\ldots,L-1\}$ such that  $\mathrm{d_H}(\mathbf v_k,\mathrm{CL}(\mathcal C))= 1$ 
  \item or for all $h\in \{1,\ldots,L\}$ the word $\mathbf v_h$ satisfies the same conditions as $\mathbf v$, that is $\mathrm{w_H}(\mathbf v_h)= \mathrm{w_H}(\mathrm{CL}(\mathbf v_h))+2$ and $\mathrm{supp_{gen}}(\mathbf v_h)[i]\neq \emptyset$ and $\mathrm{supp_{gen}}(\mathbf v_h^\prime)[i]=\emptyset$, for all $\mathbf v_h^\prime\in \mathrm{CL}(\mathbf v_h)$.
\end{enumerate} 

{ In} the second case it can be proved that $\mathrm{d_H}(\mathbf v_0,\mathrm{CL}(\mathbf v_0))=1$; therefore  in both cases there exists an element $k\in \{0,\ldots,L-1\}$ such that  $\mathrm{d_H}(\mathbf v_k,\mathrm{CL}(\mathbf v_k))=1$. Thus since $\mathbf v_k\prec \mathbf w$ and taking into account the induction hypothesis we have that $\mathbf v_k\in \tt{List}$. By Criterion 2 in Definition~\ref{List-Def-LC}, $\mathbf v_{k+1}\in \tt{List}$ and hence with a successive application of Criterion 3 we have that  $\mathbf w\in \tt{List}$. 
\end{proof}

\subsection{Leader codewords}
\begin{definition}
\label{LCwords}
The set of \emph{leader codewords} of a linear code $\mathcal C$  is defined as
\begin{equation}\nonumber
\mathrm L(\mathcal C) =\left\{ 
\begin{array}{c}
\mathbf n_1 + \mathbf e_{ij} - \mathbf n_2  \in \mathcal C\setminus \{\mathbf 0 \} \mid 
\mathbf n_1 + \mathbf e_{ij} \in \mathrm{SF}(\mathbb F_q^n),\\ \mathbf n_2 \in \mathrm{CL}(\mathcal C),
\mathrm{d_H}(\mathbf n_1,\mathrm{CL}(\mathcal C))\leq 1\\ \mbox{ and }\mathrm{d_H}(\mathbf n_1 + \mathbf e_{ij}, \mathrm{CL}(\mathcal C))\leq 1
\end{array}
\right\}.
\end{equation}
\end{definition}

Note that the definition is a bit more complex that the one for binary codes in  \cite{borges:2014} due to the fact that in the general case not all coset leaders need to be ancestors of coset leaders. The name of leader codewords comes from the fact that one could compute all coset leaders of a corresponding word knowing the set $\mathrm L(\mathcal C)$  (see \cite[Algorithm~3]{borges:2014} that can be translated straightfoward to the $q$-ary case).

\begin{theorem}[Properties of  $\mathrm L(\mathcal C)$]
\label{t:LCodW}
Let $\mathcal C$ be a linear code then
\begin{enumerate}
\item $\mathrm L(\mathcal C)$ is a test set for $\mathcal C$.
\item Let $\mathbf w$ be an element in  $\mathrm L(\mathcal C)$ then 
$$\mathrm{w_H}(\mathbf w) \leq 2 \rho(\mathcal C) +1$$ where $\rho(\mathcal C)$ is the covering radius of the code $\mathcal C$.
\item If $\mathbf w \in \mathrm L(\mathcal C)$ then $$ \mathcal X (\mathrm D(\mathbf 0)) \cap (\mathrm D(\mathbf w)\cup \mathcal X \mathrm (D(\mathbf w)))\neq \emptyset.$$
\item If $\mathcal X (\mathrm D (\mathbf 0)) \cap \mathrm D(\mathbf w) \neq \emptyset$ then $\mathbf w \in \mathrm L(\mathcal C)$.
\end{enumerate}
\end{theorem}

\begin{proof}
Items 1) and 2) follow directly from the definition of leader codewords and the proof of the same results {in the binary case} (see~\cite{borges:2014}).

{ 3)} Let $\mathbf w\in \mathrm L(\mathcal C)$, then $\mathbf w=\mathbf n_1 + \mathbf e_{ij} - \mathbf n_2$, where $\mathbf n_1,\, \mathbf n_2$ are elements in $\mathbb F_q^n$ such that $\mathbf n_1 + \mathbf e_{ij} \in \mathrm{SF}(\mathbb F_q^n)$, $\mathbf n_2 \in \mathrm{CL}(\mathcal C)$, $\mathrm{d_H}(\mathbf n_1,\mathrm{CL}(\mathcal C))\leq 1$ and $\mathrm{d_H}(\mathbf n_1 + \mathbf e_{ij},\mathrm{CL}(\mathcal C))\leq 1$. 

\begin{itemize}
  \item If $\mathbf n_1 + \mathbf e_{ij} \notin \mathrm{CL}(\mathcal C)$, then $\mathbf n_1 + \mathbf e_{ij} \in \mathcal X (\mathrm D(\mathbf 0))$ and $(\mathbf n_1 + \mathbf e_{ij}) -\mathbf w =\mathbf n_2\in \mathrm{CL}(\mathcal C)$ implies that $\mathbf n_1 + \mathbf e_{ij} \in \mathrm D(\mathbf w)$. 

\item If $\mathbf n_1 + \mathbf e_{ij} \in \mathrm{CL}(\mathcal C)$ we define  ${\mathbf n_1^\prime}={\mathbf n_1^\prime}_0=\mathbf n_1 + \mathbf e_{ij}$. It is clear that  ${\mathbf n_1^\prime},\,\mathbf n_2\in \mathrm{CL}(\mathbf n_2)$. Since $\mathbf w\neq \mathbf 0$ let $l$ be {a number in the set}  $\{1,\ldots,  n\}$ such that
{$\mathbf n_1^\prime[l] - \mathbf n_2[l]\neq 0$}. Let $\mathbf n_2[l]=\sum_{j=1}^T\,\mathbf e_{li_j}$, for $1\leq h\leq T$, $\mathbf {n_1^\prime}_h=\mathbf n_1^\prime + \sum_{j=1}^h\,\mathbf e_{li_j}$ and ${\mathbf n_2}_h=\mathbf n_2 - \sum_{j=1}^h\,\mathbf e_{li_j}$. 

If there exists an $h$ ($1\leq h<T$) such that ${\mathbf n_1^\prime}_h\notin \mathrm{CL}(\mathcal C)$ and ${\mathbf n_1^\prime}_{h-1}\in \mathrm{CL}(\mathcal C)$ then these two conditions imply that 
${\mathbf n_1^\prime}_h\in \mathcal X (\mathrm D(\mathbf 0))$ and  that ${\mathbf n_2}_h= {\mathbf n_1^\prime}_h - \mathbf w$ is either  a coset leader (${\mathbf n_1^\prime}_h\in \mathrm D(\mathbf w)$) or  $\mathrm{d_H}({\mathbf n_2}_h,{\mathbf n_2})=1$ (${\mathbf n_1^\prime}_h\in \mathcal X (\mathrm D(\mathbf w))$). 

If there is no such and $h$ ($1\leq h<T$) satisfying the condition then $\mathrm{w_H}({\mathbf n_1^\prime}_T) = {\mathrm w_H}({\mathbf n_2}_T)+1$, which means that ${\mathbf n_1^\prime}_T$ is not a coset leader and ${\mathbf n_1^\prime}_{T-1}$ is a coset leader. Then using the same idea of the previous paragraph we have that ${\mathbf n_1^\prime}_T\in \mathcal X (\mathrm D(\mathbf 0))$ and ${\mathbf n_1^\prime}_T\in \mathrm D(\mathbf w)\cup \mathcal X (\mathrm D(\mathbf w))$. 
\end{itemize}

{4)} If $\mathcal X (\mathrm D (\mathbf 0)) \cap \mathrm D(\mathbf w) \neq \emptyset$, let $\mathbf v\in \mathcal X (\mathrm D (\mathbf 0)) \cap \mathrm D(\mathbf w)$. The first condition  $\mathbf v\in \mathcal X (\mathrm D (\mathbf 0))$ implies $\mathbf v=\mathbf n_1 +\mathbf e_{ij}$ for some $\mathbf n_1$ such that $\mathrm{d_H}(\mathbf n_1,\mathrm{CL}(\mathcal C))\leq 1$ and $(i,j)\in \mathrm{supp_{gen}}(\mathbf v)$. On the other hand, $\mathbf n_1+\mathbf e_{ij} \in \mathrm D(\mathbf w)$ implies that $\mathbf n_2=(\mathbf n_1 + \mathbf e_{ij}) - \mathbf w\in \mathrm{CL}(\mathcal C)$. Therefore, $\mathbf w= \mathbf n_1 + \mathbf e_{ij}- \mathbf n_2\in \mathrm L(\mathcal C)$. 
\end{proof} 

\begin{remark}
Note that item {3)} in Theorem~\ref{t:LCodW} implies that any leader codeword is a zero neighbour however, this is one of the differences with the binary case, it is not always true that for a leader codewords $\mathbf w$ we have that $\mathcal X (\mathrm D (\mathbf 0)) \cap \mathrm D(\mathbf w) \neq \emptyset$, although by item {\it 4} we have that $\mathbf w$ is a leader codeword provided this condition is satisfied. Furthermore, {\it 4} guarantees that the set of leader codewords contains all the minimal test set according to its cardinality (see \cite{barg:1998}). As a conclusion we could say because of all this properties in Theorem~\ref{t:LCodW} that the the set of leader codewords is a ``good enought" subset  of the set of zero neightbours.
\end{remark}

\section{Correctable and uncorrectable errors}
\label{Tset-LCW}

We define the relation $\subset_1$ in the additive monoid which describe exactly the relation $\subset$ in the vector space $\mathbb F_q^n$. Given $\mathbf x,\mathbf y\in (\mathbb F_q^n,+)$
\begin{equation}
\label{eq:subset1}
 \mathbf x\subset_1 \mathbf y \hbox{ if } \mathbf x\subset \mathbf y \hbox{ and } \mathrm{supp_{gen}}(\mathbf x)\cap \mathrm{supp_{gen}}(\mathbf y-\mathbf x)=\emptyset.
\end{equation}

Note that this definition translates the binary case situation in \cite{klove}. In this case given a $\mathbf y\in (\mathbb F_q^n,+)$ { there are more words $x\in (\mathbb F_q^n,+)$  such that $\mathbf x\subset \mathbf y$ than if we consider $\mathbf x, \mathbf y$ as elements in the vector space $\mathbb F_q^n$}. Of course, any relation $\mathbf x\subset \mathbf y$ in $\mathbb F_q^n$ as a vector space is also true in the additive monoid, {but it is not true the other way round}.

The set ${E}^0({\mathcal C})$ of \textit{correctable errors} of a linear code $\mathcal C$ is the set of  the minimal { elements with respect to}  $\prec$ in each coset. The elements of the set ${E}^1({\mathcal C})=\mathbb F_q^n\setminus {E}^0({\mathcal C})$ will be  called \textit{uncorrectable errors}. A {\em trial set} $T\subset {\mathcal C}\setminus {\mathbf 0}$ of the code ${\mathcal C}$ is a set which has the following property
\begin{equation*}
\label{eq:trialset}
{\mathbf y}\in {E}^0({\mathcal C}) \mbox{ if and only if } {\mathbf y} \leq {\mathbf y} + {\mathbf c}, \mbox{ for all } {\mathbf c}\in T.\end{equation*}

Since $\prec$ is a monotone $\alpha$-ordering on $\mathbb F_q^n$, the set of correctable and uncorrectable errors form a monotone structure. Namely, if $\mathbf x\subset_1 \mathbf y$ then $\mathbf x \in {E}^1({\mathcal C})$ implies $\mathbf y \in {E}^1({\mathcal C})$ and $\mathbf y \in {E}^0({\mathcal C})$ implies $\mathbf x \in {E}^0({\mathcal C})$. In the {general case $q\neq 2$ there is} a  difference with respect to the binary case, there may be words $\mathbf y^\prime\in \mathbb F_q^n$ s.t. $\mathrm{supp_{gen}}(\mathbf y^\prime)=\mathrm{supp_{gen}}(\mathbf y)$, {$\mathbf y^\prime\subset \mathbf y$} and $\mathbf y^\prime$ could be either a correctable error or an uncorrectable error, so, the monotone structure it is not sustained by $\subset$ in the additive monoid $(\mathbb F_q^n,+)$. Note that in Definition~\ref{List-Def-LC}, the set denoted by ${\cal N}$ is the set ${E}^0({\mathcal C})$ corresponding to $\prec$.

Let the set of minimal uncorrectable errors $M^1({\mathcal C})$ be the set of $\mathbf y\in {E}^1({\mathcal C})$ such that, if $\mathbf x \subseteq_1 \mathbf y$ and $\mathbf x\in {E}^1({\mathcal C})$, then $\mathbf x= \mathbf y$. In a similar way, the set of maximal correctable errors is the set $M^0({\mathcal C})$ of elements $\mathbf x\in {E}^0({\mathcal C})$ such that, if $\mathbf x\subseteq_1 \mathbf y$ and $\mathbf y\in {E}^0({\mathcal C})$, then $\mathbf x= \mathbf y$.

For $\mathbf c\in {\mathcal C}\setminus \mathbf 0$, a {\em larger half} is defined as a minimal word $\mathbf u$ in the ordering $\preceq$ such that $\mathbf u - \mathbf c \prec \mathbf u$. The weight of such a word $\mathbf u$ is such that $$\mathrm{w_H}(\mathbf c)\leq 2\mathrm{w_H}(\mathbf u)\leq \mathrm{w}_H(\mathbf c)+ 2,$$ see \cite{klove} for more details. The set of larger halves for a codeword $\mathbf c$ is denoted by $\mathrm{L}(\mathbf c)$, and for $U\subseteq {\mathcal C}\setminus \mathbf 0$ the set of larger halves for elements of $U$ is denoted by $\mathrm{L}(U)$. Note that $\mathrm{L}({\mathcal C})\subseteq {E}^1({\mathcal C})$.

For any ${\mathbf y\in \mathbb F_q^n}$, let $H(\mathbf y)=\{c\in {\mathcal C:\, \mathbf y- \mathbf c \prec \mathbf y}\}$, and we have $\mathbf y\in {E}^0({\mathcal C})$ if and only if $H(\mathbf y)=\emptyset$, and $\mathbf y\in {E}^1({\mathcal C})$ if and only if $H(\mathbf y)\neq \emptyset$. 

In   \cite[Theorem~1]{klove} there is a  characterization of the set $M^1({\mathcal C})$ in terms of $H(\cdot)$ and larger halves of the set of minimal codewords $M({\mathcal C})$ for the binary case. It is easy to proof that this Theorem and    \cite[Corollary~3]{klove} are also true for any linear code.

\begin{proposition}[Corollary 3 in \cite{klove} ]
\label{prop:klove}
Let ${\mathcal C}$ be a linear code and $T\subseteq {\mathcal C}\setminus \mathbf 0$. The following statements are equivalen,:
\begin{enumerate}
\item $T$ is a trial set for ${\mathcal C}$.
\item If $\mathbf y \in M^1({\mathcal C})$, then $T\cap H(\mathbf y)\neq \emptyset$.
\item $M^1({\mathcal C})\subseteq \mathrm{L}(T)$.
\end{enumerate}
\end{proposition}

Now we will formulate the result which relates the trial sets {for a given weight compatible order $\prec$} and the set of leader codewords.
\begin{theorem}
\label{teo:ts-LCW}
Let $\mathcal C$ be a linear code and $\mathrm L(\mathcal C)$ the set of leaders codewords for $\mathcal C$, the following statements are satisfied.
\begin{enumerate}
\item $\mathrm L(\mathcal C)$ is a trial set for any given $\prec$.
\item $\mathrm L(\mathcal C)$ contains any minimal trial set for any given $\prec$.
\item Algorithm~2 in \cite{borges:2014} can be adapted to compute a set of leader codewords which is a trial set $T$ for a given $\prec$ such that  it satisfies the following property
\begin{description}
\item[] For any $\mathbf c\in T$, there exists $\mathbf y\in M^1({\mathcal C})\cap \mathrm{L}(\mathbf c)$ s.t. $\mathbf y- \mathbf c\in {E}^0({\mathcal C})$.
\end{description}
\end{enumerate}
\end{theorem}
\begin{proof}$\,$\\
{Proof of 1) }We will prove statement 2 of Proposition~\ref{prop:klove}. Let $\mathbf y \in M^1({\mathcal C})$, let $i$ such that $\mathrm{supp_{gen}}(\mathbf y)[i]\neq \emptyset$ and $\mathbf n_1=\mathbf y -\mathbf y_i$. Since $\mathbf y \in M^1({\mathcal C})$ we have that $\mathbf n_1\in {E}^0({\mathcal C})$, thus it is also a coset leader. On the other hand, let $\mathbf n_2\in {E}^0({\mathcal C})$ such that $\mathbf n_2 \in \mathrm{CL}(\mathbf y)$ and $\mathbf c=\mathbf y - \mathbf n_2$. It is clear that $\mathbf c$ is a leader codeword and $\mathbf y - \mathbf c =\mathbf n_2\prec \mathbf y$. Therefore  $\mathbf c\in H(\mathbf y)$.

\noindent {Proof of 2) }By statement 3 of Proposition~\ref{prop:klove}, a minimal trial set should contain at least as larger halves the elements in the set $M^1({\mathcal C})$. But it can be seen that a set of codewords having as larger halves at least the set $M^1({\mathcal C})$ is already a trial set.

Let us show that any codeword with a larger half belonging to $M^1({\mathcal C})$ is a leader codeword. Let $\mathbf c\in {\mathcal C}$ such that $\mathrm{L}(\mathbf c)\cap M^1({\mathcal C})\neq \emptyset$. Then there exists an element $\mathbf y\in M^1({\mathcal C})$ such that $\mathbf y\in \mathrm{L}(\mathbf c)$. Thus $\mathbf c= \mathbf y-\mathbf n_2$ with $\mathbf n_2 \prec \mathbf y$. On the other hand let $i\in \{1,\ldots,n\}$ such that $\mathrm{supp_{gen}}(\mathbf y)[i]\neq \emptyset$, and let $\mathbf n_1=\mathbf y - \mathbf y_i$. Then  $\mathbf y\in M^1({\mathcal C})$ implies that $\mathbf n_1\in {E}^0({\mathcal C})$ and therefore it is a coset leader. from the fact $\mathbf n_2\prec \mathbf y$ we have two cases,

\begin{description}
  \item[i.-] $\mathrm w_H(\mathbf y)>\mathrm w_H(\mathbf n_2)$, in this case since $\mathbf n_1\in \mathrm{CL}(\mathcal C)$ we have that $\mathbf n_2\in \mathrm{CL}(\mathcal C)$ and then $\mathbf c$ is a leader codeword.

\item[ii.-] $\mathrm w_H(\mathbf y)=\mathrm w_H(\mathbf n_2)$, then $d_H(-\mathbf n_2+\mathbf y_i,\mathrm{CL}(\mathcal C))\leq 1$, then, $\mathbf c=-\mathbf n_2+\mathbf y_i+\mathbf n_1$ is a leader codeword.

\end{description}
\noindent {Proof of 3) }In the steps of the construction of the leader codewords in Algorithm 2 (Steps 11 - 13 for the binary case) it is enough to state the condition $\mathbf t\in M^1({\mathcal C})$ for $\mathbf t$ and take $\mathbf t_k$ only equal to the coset leader of $\mathrm{CL}(\mathbf t)$, that is the corresponding correctable error and add the codeword $\mathbf t-\mathbf t_k$ to the set $\mathrm L(\mathcal C)$. 

\end{proof}

{{\bf Acknowlegments}.  M. Borges-Quintana has been supported by a Post-doctorate scholarship at the University of Valladolid (09-2014 to 02-2015) by Erasmus Mundus Program, Mundus Lindo Project.} E. Mart\'inez-Moro is partially funded by Spanish MinEco research grant  
MTM2012-36917-C03-02.

%





\end{document}